\documentclass[runningheads]{llncs}

\usepackage{amsmath}

\usepackage{courier}            
\usepackage[scaled]{helvet} 
\usepackage{url}                  
\usepackage{listings}          
\usepackage{enumitem}      
\usepackage[colorlinks=true,allcolors=blue,breaklinks,draft=false]{hyperref}   

\usepackage{xspace}
\usepackage{proof}
\usepackage{graphicx}
\usepackage{pdfpages}
\usepackage{amsmath}
\usepackage{amsthm}
\usepackage{amscd}
\usepackage{amssymb}
\usepackage{float}
\usepackage{epsfig}
\usepackage{xcolor}
\usepackage{color}
\usepackage{subcaption}
\usepackage{filecontents}
\usepackage{moresize}

\usepackage{longtable, tabu}

\usepackage[T1]{fontenc}
 
\newcommand{\hide}[1]{} 

\newcommand{\etal}{\textit{et al.}\@\xspace}

\numberwithin{equation}{section}

\usepackage[T1]{fontenc}

\definecolor{dkgreen}{rgb}{0,0.3,0}
\definecolor{gray}{rgb}{0.5,0.5,0.5}
\definecolor{mauve}{rgb}{0.58,0,0.82}
\definecolor{light-gray}{gray}{0.80}

\usepackage{listings}
\lstdefinelanguage{spmd}{
  morekeywords = {
       loc, bit, bool, true, false
     , implements, harness
     , null
     , assert, assume
     , else
     , find, fix, fold, for, forall, function
     , generator, gen
     , if, while, int, float, bool, string
     , loop, simple, cond, val
     , fork, join
     , nil, null, none, new, malloc
     , option, or
     , ref, return
     , spmdfork, nprocs, spmdtransfer
     , void
     , concrete, sym
     , requires, ensures
     , invariant, decreases 
     , conj, exp
     , init, stmt},
  literate=
    {-}{--}1,
  morecomment=[l]{//}
}
\renewcommand{\scriptsize}{\fontsize{8.5}{9}\selectfont}
\makeatletter
\lst@AddToHook{TextStyle}{\let\lst@basicstyle\scriptsize\fontfamily{lmss}\selectfont}
\makeatother

\usepackage{url}

\usepackage{rotating}

\usepackage{enumitem}

\usepackage[linesnumbered]{algorithm2e}
\usepackage{gastex}
\usepackage{multirow}

\newcommand{\conf}[1]{}

\newcommand{\name}{\textsc{ConcSolver}~}

\newcommand{\undefined}{{\textit{ndef}}}

\newcommand{\pre}{\textit{pre}}
\newcommand{\post}{\textit{post}}

\newcommand{\trans}{\textit{\textbf{trans}}}
\newcommand{\inv}{\textit{inv}}

\newcommand{\ite}{\textsf{ite}}

\newcommand{\vecx}{\vec{x}}
\newcommand{\vecy}{\vec{y}}
\newcommand{\vectt}{\vec{t}}
\newcommand{\vecc}{\vec{c}}
\newcommand{\fast}{\textit{\textbf{fast-trans}}}
\newcommand{\branches}{\textit{Brchs}}

\usepackage{longtable}

\urldef{\mailsa}\path|{xkqiu, asolar}@csail.mit.edu|

\begin{document}

\mainmatter  

\title{Reconciling Enumerative and Symbolic Search in Syntax-Guided Synthesis}


%
%
\author{Kangjing Huang\inst{1}%
\and Xiaokang Qiu\inst{1}
\and Qi Tian\inst{2}
\and Yanjun Wang\inst{1}}
%

\institute{Purdue University \and Nanjing University
}

%
%

\toctitle{Lecture Notes in Computer Science}
\tocauthor{Authors' Instructions}
\maketitle

\begin{abstract}
Syntax-guided synthesis aims to find a program satisfying semantic specification as well as user-provided structural hypothesis. For syntax-guided synthesis there are two main search strategies: \emph{concrete search}, which systematically or stochastically enumerates all possible solutions, and \emph{symbolic search}, which interacts with a constraint solver to solve the synthesis problem. In this paper, we propose a concolic synthesis framework which combines the best of the two worlds. Based on a decision tree representation, our framework works by enumerating tree heights from the smallest possible one to larger ones. For each fixed height, the framework symbolically searches a solution through the counterexample-guided inductive synthesis approach. To compensate the exponential blow-up problem with the concolic synthesis framework, we identify two fragments of synthesis problems and develop purely symbolic and more efficient procedures. The two fragments are decidable as these procedures are terminating and complete. We implemented our synthesis procedures and compared with state-of-the-art synthesizers on a range of benchmarks. Experiments show that our algorithms are promising.
\end{abstract}

\section{Introduction}
\label{sec:intro}

Syntax-guided synthesis (SyGuS) is a common theme underlying many program synthesis systems. 
The insight behind SyGuS is that to synthesize a large-scale program automatically, the user needs to provide not only a semantic specification but also a syntactic specification, i.e., a grammar of candidate programs as the search space.
SyGuS has seen great success in the last decade, including the Sketch~\cite{sketch,sketchthesis,sketch:manual} synthesizer and the Flash~Fill feature of Microsoft Excel~\cite{flashfill,flashfill2}. The research community has also developed a standard interchange format for SyGuS problems and organized an annual competition, which encourages a plethora of syntax-guided synthesizers~\cite{sygus,sygus17}.

There is a dichotomy of existing synthesizers based on their underlying search strategies: \emph{concrete (enumerative) search} --- which systematically enumerates all possible implementations of the function to be synthesized and check if it satisfies the desired specification; and \emph{symbolic search} --- which solves the synthesis problem through a synthesis procedure interacting with a SAT/SMT solver. Neither strategy clearly outperforms the other. Concrete search typically begins the search from smaller-sized candidates, which usually can be checked more efficiently, and move toward larger-sized candidates. This strategy guarantees to produce the smallest program, if any satisfying program exists, and is proven efficient for a wide spectrum of syntax-guided synthesis tasks. For example, EUSolver~\cite{Alur2017} has been the winner of the general track in the past two years' SyGuS competition~\cite{sygus16,sygus17}. In contrast, symbolic search leverages ad hoc synthesis procedures for specific classes of synthesis problems, e.g., integer arithmetic or bit vectors, and is capable of synthesizing large programs more efficiently. For example, CVC4~\cite{Reynolds2015} has won the CLIA track of SyGuS competition three years in a row.

In this paper, we present a \emph{concolic} SyGuS synthesis framework that breaks the dichotomy and reconciles the concrete and symbolic search strategies. The framework represents programs as decision trees and enumerates all possible tree heights, from smaller to larger. For each fixed height $h$, the synthesis procedure symbolically searches whether there is a decision tree of height $h$ satisfying the specification. If a decision tree is found at height $h$, the procedure terminates and returns the found program; otherwise the procedure moves to height $h+1$ and continues the symbolic search. This concolic framework combines the advantages of two worlds: on the one hand, it still enumerates from simpler subtasks and guarantees to synthesize the smallest satisfying program; on the other hand, the framework can leverage efficient symbolic search algorithms at each height to accelerate the enumeration process. In this paper, we focus on \emph{Conditional Linear Integer Arithmetic} (CLIA) and elaborate how the concolic search strategy can be integrated into a Counterexample-Guided Inductive Synthesis (CEGIS) framework. We believe the idea of height/size-based enumeration and the integration with CEGIS are general enough and may lead to new synthesis algorithms for other classes of SyGuS problems.

A common problem faced by enumerative search approaches, including our concolic synthesis framework, is their scalability in program size. When the smallest solution is large-sized, the enumerative synthesizer may take drastically long time to discard all smaller-sized unsatisfying programs. To address this problem, we complement our synthesis framework with two \emph{decidable fragments}, namely \emph{Strong Single Invocation} (SSI) and \emph{Acyclic Translational} (AT) invariant synthesis. 
When the input synthesis problem falls into one of these decidable fragments, our framework invokes the corresponding synthesis procedure to solve the synthesis problem purely symbolically. 

This paper makes the following contributions: a) proposes a concolic syntax-guided synthesis framework that combines concrete search and symbolic search synergistically; b) identifies two decidable fragments of synthesis problems and presents their terminating and complete synthesis procedures; c) compares the new algorithms with existing synthesizers on 166 benchmarks and demonstrates the effectiveness and efficiency of the concolic synthesis approach.

%

\section{Concolic Synthesis}
\label{sec:concolic}

In this section, we present our concolic synthesis framework for synthesizing conditional linear integer arithmetic functions.

\subsection{The CLIA Synthesis Problem}

We first present the class of synthesis problems we tackle in this paper: Conditional Linear Integer Arithmetic (CLIA). The CLIA synthesis problem is probably the most well studied class of synthesis problems and embodied in a lot of common synthesis tasks such as sketch-based program synthesis or loop invariant synthesis.

Overall, a CLIA synthesis problem can be represented as a formula $\exists f \forall \vecx \varphi(f; \vecx)$, where $f$ is the name of the function to be synthesized, $\vecx$ is a vector of variables, and $\varphi(f; \vecx)$ is a formula for which the grammar is shown in Figure~\ref{fig:syntax}. A solution to the synthesis problem is a term $t(\vecy)$ from the same grammar such that replacing $f$ with $t(\vecy)$ in $\varphi(f; \vecx)$ makes the formula valid in the theory of linear arithmetic. Formally, $\forall \vecx \varphi\big(\lambda \vecy.t(\vecy); \vecx\big)$ should be a valid LIA formula.

Note that the function to be synthesized can return not only integer values but also boolean values. In other words, the CLIA synthesis problem we handle covers not only the CLIA track but also the INV track of the SyGuS competition.

\begin{figure}
\begin{displaymath}
\begin{array}{rll}
\multicolumn{3}{l}{~~~~~~f:~\textrm{Function Name}  ~~~~~~~~~~~~~\vec{t}:~\textrm{Term Vector}} \\
\textrm{~Term:~} t, t_1, t_2 & ::= & 0 ~\big|~ 1 ~\big|~ x ~\big|~ t_1 + t_2 ~\big|~ \ite(c, t_1, t_2) ~\big|~ f(\vec{t}) \\
\textrm{~Condition:~} \varphi, \varphi_1, \varphi_2 & ::= & t_1 \geq t_2 ~\big|~ \varphi_1 \wedge \varphi_2 ~\big|~ \neg \varphi ~\big|~ f(\vec{t}) \\
\end{array}
\end{displaymath}
\caption{Syntax of CLIA Synthesis}\label{fig:syntax}
\end{figure}

\subsection{Decision Tree Representation}

To synthesize a function $f(\vecx)$ in CLIA, we assume the implementation of $f$ is represented in the \emph{decision tree normal form}, as described in Figure~\ref{fig:dtnf}. It is not hard to see that every CLIA term can be converted to this normal form. The proof relies on the fact that every atomic formula $t_1 \geq t_2$ can be rewritten in the form of $\vecc \cdot \vecx + d \geq 0$, where $\vecc$ is a vector of constants and $\vecx$ is a vector of variables. Then the decision tree representation of a CLIA term is a binary tree in which every node with id $i$ contains a vector $\vecc_i$ of integer constants and an extra constant $d_i$. Then each decision node (non-leaf node) tests whether $\vecc_i \cdot \vecx + d_i \geq 0$ and according to the test result proceeds to the ``true'' child or ``false'' child. Each leaf node determines the value of the function as $\vecc_i \cdot \vecx + d_i$. For example, the binary max function $$max2(x_1, x_2) \stackrel{\textit{def}}{=} \ite(x_1 \geq x_2, x_1, x_2)$$ can be represented as the tree shown in Figure~\ref{fig:max2}.

Representing invariants is similar. The only difference is that the leaf node will return true or false, depending on whether the associated expression is non-negative or not.

We assume that the decision-tree is a full binary tree -- if it is not full, one can extend the tree with extra nodes with a tautology condition, e.g., $1 \geq 0$. Notice that the full decision tree of height $h$ consists of $2^h - 1$ nodes, and the node id's can be fixed in the range between $0$ and $2^h-2$. This allows us to reduce the SyGuS synthesis problem to the problem of searching for:
\begin{itemize}
\item[a)] the height of the decision tree; 
\item[b)] the values $\vecc_i$ and $d_i$ for each node $i$. 
\end{itemize}

\begin{figure}
\begin{displaymath}
\begin{array}{rll}
\multicolumn{3}{l}{\textrm{~~~~~~Int Const Vector:~}  \vecc_i } \\
\multicolumn{3}{l}{\textrm{~~~~~~Int Const:~}  d_i } \\
\textrm{~Atom Expr:~} e & ::= & \displaystyle \vecc_i \cdot x + d_i \\
\textrm{~Atom Cond:~} \alpha & ::= & e \geq 0 \\
\textrm{~Expr:~} E, E_1, E_2 & ::= & \displaystyle e ~\big|~ \ite (\alpha, E_1, E_2) \\
\textrm{~Condition:~} \varphi & ::= & \alpha ~\big|~ \ite (\alpha, \varphi_1, \varphi_2) \\
\end{array}
\end{displaymath}
\caption{Decision Tree Normal Form}\label{fig:dtnf}
\end{figure}

\begin{figure}
\begin{center}
\unitlength=2mm
\includegraphics{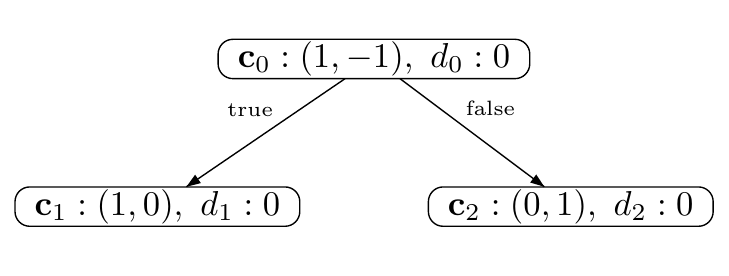}

\end{center}
\caption{Representation of the $max2$ function}\label{fig:max2}
\end{figure}

\subsection{Concolic Search}

\begin{algorithm}
\SetKwInOut{Input}{input}
\SetKwInOut{Output}{output}
\SetKwProg{Fn}{def}{\string:}{}
\SetKwFunction{Conc}{concolic-synth}
\SetKwFunction{Fixed}{fixed-height-synth}
\SetKwFunction{CheckSat}{check-sat}
\SetKwFunction{Synth}{synthesize}
\SetKwFunction{Tree}{checkTreeness}
\SetKwFor{For}{for}{:}{}%
\SetKwFor{ForEach}{foreach}{}{}
\SetKwIF{If}{ElseIf}{Else}{if}{:}{elif}{else:}{}%
\SetKwFor{While}{while}{:}{fintq}%
\SetKw{Assert}{assert}
\SetKw{True}{true}
\SetKw{False}{false}
\AlgoDontDisplayBlockMarkers
\SetAlgoNoEnd
\SetAlgoNoLine%
\Input{A specification $\varphi(f; \vecx)$}
\Output{A term $t(\vecy)$ satisfying $\varphi(\lambda \vecy. t(\vecy); \vecx)$ }
$E \leftarrow \emptyset$ \\
\Fn{\Conc{$\varphi(f; \vecx)$}}{
  $h \leftarrow 1$ \\
  \Repeat{$t \neq \bot$}{
    $t \leftarrow$ \Fixed{$\varphi(f; \vecx)$, $h$} \\
    $h \leftarrow h+1$
  }
  \Return{$t$} \\
}
\Fn{\Fixed{$\varphi(f; \vecx)$, $h$}}{
  $candidate \leftarrow 0$ \\
  \Repeat{$t = \bot$}{
    $solution \leftarrow$ \CheckSat{$\neg \varphi(\lambda \vecy.candidate(\vecy); \vecx)$} \\
    \If{$solution =$ \textsf{unsat}}{\Return{$solution$}}
    \Else{$E \leftarrow E \cup \{solution\}$; $t(\vecy) \leftarrow$ \Synth{$\displaystyle \bigwedge_{\vec{e} \in E} \varphi(f; \vec{e})$}}
  }
  \Return{$t$} \\
}

\caption{Concolic synthesis algorithm}
\label{alg:concolic}
\vspace{.1in}
\end{algorithm}

Overall, our concolic search algorithm, presented as Algorithm~\ref{alg:concolic}, adopts the standard CEGIS framework~\cite{sketch}. In the setting of our CLIA synthesis, the CEGIS framework consists of two parts: a verifier and a synthesizer. The framework works by maintaining a set of counterexamples $E$, which is a set of integer vectors. In each iteration, the synthesizer proposes a candidate solution $t(\vecy)$ that satisfies the specification when $\vecx$ is assigned values from $E$, i.e., $\displaystyle \bigwedge_{\vec{e} \in E} \varphi(\lambda \vecy.t(\vecy); \vec{e})$. Then the verifier checks whether the candidate satisfies the specification: $\forall \vecx \varphi(\lambda \vecy.t(\vecy); \vecx)$. If yes, then $f$ is the desired function and the algorithm terminates; otherwise the verifier reports a new vector $\vec{e}_c$ as the witness of the failed verification, i.e., $\neg \varphi(\lambda \vecy.t(\vecy); \vec{e}_c)$, and add $\vec{e}_c$ to $E$. Then CEGIS continues with the next iteration repeatedly, until a solution is found.

In our framework, the verifier can be easily constructed. Given a candidate solution $f(\vecy) = t(\vecy)$, $f$ is the desired function if and only if the negation of the specification $\exists \vec{x}. \neg \varphi(\lambda \vecy.t(\vecy); \vecx)$ is unsatisfiable. The unsatisfiability of QF\_LIA can be checked by an SMT solver (\texttt{check-sat} in line 12).

The synthesizer part is the crux of the concolic synthesis framework. It enumerates all possible heights of the decision tree, starting from $1$ (the loop in lines 4--7). In each iteration of the CEGIS loop, the synthesizer attempts to synthesize a solution with the current height $h$ (\texttt{fixed-height-synth} in line 5). If there is no solution, the synthesizer increases the height to $h+1$ and retry.

Now in each iteration, the synthesis task is to synthesize a term $t(\vecy)$ of a fixed height $h$ such that $\displaystyle \bigwedge_{\vec{e} \in E} \varphi(\lambda \vecy.t(\vecy); \vec{e})$. 
This task can also be encoded as an SMT query (\texttt{synthesize} in line 16). Notice that in $\varphi(f; \vec{e})$, $f$ is applied to constants only and in the form of $f(\vecc)$ where $\vecc$ is an integer vector. For example, if an assignment $\vec{e} = (1,2)$ is a counterexample in $E$ and the current height $h=2$, then the expression $f(1, -2)$ can be represented as an LIA expression: $$\ite \big(\vecc_0 \cdot (1,-2) + d_0 \geq 0,~~ \vecc_1 \cdot (1,-2) + d_1,~~ \vecc_2 \cdot (1,-2) + d_2\big)$$ Then the SMT query is simply a conjunction that enumerates all $\vec{e} \in E$ and replace $f(\vec{e})$ with the corresponding encoding LIA expression.

\subsection{Parallelization}

While the naive concolic synthesis algorithm incrementally searches all possible height of the decision tree, starting from $1$, the algorithm can be naturally parallelized. If there are $n$ cores available on the machine, the parallelized concolic synthesis framework runs the same CEGIS algorithm with $n$ different heights on the $n$ cores independently, i.e., each maintains a separate set of counterexamples. The algorithm starts with the $n$ smallest heights, $\{1, \dots, n\}$. It also maintains a variable $p$ as the next height to search, starting from $n+1$. Whenever a core concludes that there is no solution at the current height, it starts a new CEGIS loop at height $p$, and the value of $p$ gets increased. The whole algorithm stops whenever a core finds a solution.

\section{Synthesis with the Strong Single Invocation Property}

As an enumerative approach, the concolic synthesis framework presented in Section~\ref{sec:concolic} works very efficiently when the smallest solution's decision tree is not high. However, the performance may slow down drastically when the smallest solution's decision tree's height increases. For example, our algorithm can synthesize the running example \textsf{max2} in $0.7$ second, then synthesize \textsf{max3} (max of three arguments) in $2.6$ seconds; but \textsf{max4} takes $876$ seconds to synthesize. The main reason is that the SMT solving time increases exponentially while the size of the decision tree (and hence the size of the formula solved in each iteration) increases linearly. Moreover, the number of counterexamples needed to determine the unsatisfiability at a given height level also blows up.

To this end, we identified two decidable fragments of CLIA synthesis, for which pure symbolic synthesis/decision procedures obviously outperform the general concolic search algorithm. In this section and the next section, we describe these two fragments, prove their decidability and present the decidable synthesis procedures. These two procedures are integrated with the concolic search algorithm as privileged alternatives: if a synthesis problem falls in one of the decidable fragments, the corresponding decision procedure will be invoked instead of the general concolic search algorithm. The user can be completely agnostic to the choice of the underlying algorithm.

\subsubsection{Strong single invocation (SSI).}

The first decidable fragment is defined based on a strong single invocation property:
\begin{definition}
A CLIA synthesis problem $\exists f \forall \vecx \varphi (f; \vecx)$ has the Strong Single Invocation (SSI) property if there is a single vector of terms $\vectt$ such that
\begin{enumerate}
\item $f$ only occurs in $\varphi$ as $f(\vectt)$;
\item in each atomic formula, $f(\vectt)$ occurs at most once.
\end{enumerate}
\end{definition}
This fragment is actually inspired by the Single Invocation (SI) property identified by Reynolds~\etal~\cite{Reynolds2015}, for which they proposed a complete but non-terminating procedure called Counterexample-Guided Quantifier Instantiation (CEGQI). SSI strengthens the SI property with the second single-occurrence property and guarantees the termination and decidability. To understand the difference, consider $\exists f \forall x (f(x) +  f(x) = 4x)$. This synthesis problem satisfies the SI property but does not satisfy the SSI property, as the second property in the definition is violated.

\begin{theorem}[Decidability of SSI]
There is a terminating synthesis procedure that takes as input a CLIA synthesis problem $\exists f \forall \vecx \varphi (f; x)$ satisfying the SSI property, either outputs a satisfying term $t(\vecy)$ as a solution exists, or concludes the nonexistence of a solution.
\end{theorem}
\begin{proof}
Due to the first SSI property, the formula $\exists f \forall \vecx \varphi (f; \vecx)$ is logically equivalent to $\forall \vecx \exists z \varphi (z; \vecx)$, where $\varphi(z; \vecx)$ replaces the all occurrences of $f(\vectt)$ in $\varphi (z; \vecx)$ with $z$.

Now notice that its negation, $\forall z \neg \varphi(z; \vecx)$, is a first-order CLIA formula, and the satisfiability can be decided. Obviously, if the formula is satisfiable, the satisfying assignment of $\vecx$ witnesses the original synthesis problem is unsynthesizable. Hence what remains is to show that if $\forall z \neg \varphi(z; \vecx)$ is unsatisfiable, i.e., $\forall \vecx \exists z \varphi(z; \vecx)$ is valid, the function $f$ can be implemented: there is a CLIA term $t(\vecx)$ such that $\varphi(t(\vecx); \vecx)$ is valid.

As $\varphi(z; \vecx)$ is quantifier-free, we can assume it is in the disjunctive normal form. Then due to the second SSI property, $\varphi(z; \vecx)$ can be written as 
\begin{equation}\label{normalized}
\displaystyle \bigvee_{i} \Big( \bigwedge_{j} \big(z \geq t_{i,j}(\vecx)\big) \wedge \bigwedge_{k} \big(z \leq u_{i,k}(\vecx)\big) \Big)
\end{equation}
Now as we assume $\forall \vecx \exists z \varphi(z; \vecx)$ is valid, there must exist an $i$ such that $t_{i,j}(\vecx)$ is not greater than $u_{i,k}(\vecx)$, for any $j,k$. Then one can simply let $z$ be the maximum value among $t_{i,j}(\vecx)$ (or the minimum value among $u_{i,k}(\vecx)$). Formally, there is a straightforward implementation of $f$:
\begin{equation}\label{straightforward-solution}
\begin{array}{lll}
\displaystyle f(\vecx) & = & \ite \Big( \big( \bigwedge_{j,k} t_{1,j}(\vecx) \leq u_{1,k}(\vecx) \big),~~ \textsf{MAX}_j(t_{1,j}(\vecx)),  \\
& & ~~~~ \ite \Big( \big( \bigwedge_{j,k} t_{2,j}(\vecx) \leq u_{2,k}(\vecx) \big),~~ \textsf{MAX}_j(t_{2,j}(\vecx)), \\
& & ~~~~~~~~ \ite \Big( \cdots ,~~ \textsf{MAX}_j(t_{3,j}(\vecx)), \\
& & ~~~~~~~~~~~~ \cdots, 0 \Big) \cdots \Big) \Big) \\
\end{array}
\end{equation}
where $\textsf{MAX}_j(t_{i,j}(\vecx))$ represents the nested $\ite$-term selecting the term $t_{i,m}(\vecx)$ with the greatest value.
\end{proof}

The straightforward implementation~\ref{straightforward-solution} requires a DNF-conversion, which results in an exponential size blowup. This is actually unnecessary and the following corollary shows a polynomial-size implementation.
\begin{corollary}
There is a synthesis procedure that constructs solutions for CLIA synthesis problems with SSI properties in polynomial time, if a solution exists.
\end{corollary}
\begin{proof}
We assume $\varphi (z; \vecx)$ is normalized such that every atomic formula is in the form of $z \geq t(\vecx)$ or $z \leq t(\vecx)$ (similar to those in Equation~\ref{normalized}). This normalization can be done in polynomial time. Let the set of terms involving $\vecx$ be $T$, then according to the straightforward solution~\ref{straightforward-solution}, for any $\vecx$, the value of $f(\vecx)$ can always be a term from $T$. Hence we can build an alternative implementation by enumerating every term $t_i(\vecx) \in T$ as the value of $f(\vecx)$ and check if $\varphi(t_i(\vecx); \vecx)$ is satisfied:
\begin{equation}\label{better-solution}
\begin{array}{lll}
\displaystyle f(\vecx) & = & \ite \Big( \lambda \vecx. \varphi(t_1(\vecx); \vecx),~ t_1(\vecx),  \\
& & ~~~~ \ite \Big( \varphi(\lambda \vecx. t_2(\vecx); \vecx),~ t_2(\vecx), \\
& & ~~~~~~~~ \ite \Big( \cdots ,~ t_3(\vecx), \\
& & ~~~~~~~~~~~~ \cdots, 0 \Big) \cdots \Big) \Big) \\
\end{array}
\end{equation}
Above is clearly a solution and can be constructed in polynomial time.
\end{proof}

\subsubsection{Extension for Commutativity.} 

While the SSI property leads to an efficient and decidable terminating procedure described above, the property can be easily violated by practical synthesis problems. For example, the user may want to specify the commutativity in addition to other properties for the desired function: $$\exists f \forall x_1,x_2: f(x_1,x_2) = f(x_2,x_1) \wedge \varphi(f;~ x_1,x_2)$$ However, with the commutativity, we only need to focus on the values of $f(x_1,x_2)$ when $x_1 \geq x_2$, because otherwise we can use the value of $f(x_2,x_1)$ instead. Hence we can synthesize a function $g$ with SSI specification first: $$\exists g \forall x_1,x_2: \big(x_1 \geq x_2 \rightarrow \varphi(g;~ x_1,x_2)\big)$$ Now $g$ can be synthesized using our SSI synthesis procedure; then the desired function $f$ can be written as $f(x_1,x_2) \stackrel{\textit{def}}{=} \ite(x_1 \geq x_2, g(x_1,x_2), g(x_2,x_1))$.

\section{Acyclic Translational Invariant Synthesis}

Another decidable fragment we identified is for invariant synthesis only. The invariant synthesis problem we consider in this paper is a special kind of CLIA synthesis problem:

\begin{definition}[Invariant synthesis problem]
An invariant synthesis problem can be represented as $\exists \inv \forall \vecx \varphi(\inv; \vecx)$, where $\inv$ is the predicate to be synthesized and $\varphi(\inv; \vecx)$ is of the form $$\varphi(\inv; \vecx) \equiv \Big( \pre(\vecx) \rightarrow \inv(\vecx) \Big) \wedge \Big( \inv(\vecx) \rightarrow \inv(\trans(\vecx)) \Big) \wedge \Big( \inv(\vecx) \rightarrow \post(\vecx) \Big)$$ where $\pre(\vecx)$ and $\post(\vecx)$ are CLIA conditions over $\vecx$ defined in the Syntax of Figure~\ref{fig:syntax}, $\trans(\vecx)$ defines a vector of CLIA terms over $\vecx$ such that $|\trans(\vecx)| = |\vecx|$.
\end{definition}

\label{sec:invariant}


Intuitively, $(\pre(\vecx),\trans(\vecx),\post(\vecx))$ represents a program with a set of variables $\vecx$. $\pre(\vecx)$ and $\post(\vecx)$ are the pre- and post-conditions, respectively. $\trans(\vecx)$ represents the iterative transition: $\vecx := \trans(\vecx)$. The loop terminates when $\trans(\vecx) = \vecx$. The goal of the synthesis problem is to find a loop invariant guaranteeing the partial correctness of the program with respect to $\pre$ and $\post$.

\subsection{Translational Invariant Synthesis}

The decidable fragment we defined for invariant synthesis focuses on a special class of \emph{translational programs}. Intuitively, a program is translational if all variable updates in the loop are of the form $x_i := x_i + c_i$ where $c_i$ is an integer constant.

\begin{definition}[Translational Invariant Synthesis]
An invariant synthesis problem is translational if the corresponding $\trans(\vecx)$ term is defined in a way such that all atomic terms occurred in it are of the form $\vecx + \vecc$.
\end{definition}

We also assume the synthesis problem is normalized as below:
\begin{enumerate}
\item in $\trans(\vecx)$, for each occurrence of conditional term $\ite(\varphi, \vectt_1, \vectt_2)$, $\varphi$ is free of disjunction, and $\vectt_1$ is an atomic term of form $\vecx + \vecc$; and for each two occurrences $\ite(\varphi, \vectt_1, \vectt_2)$ and $\ite(\psi, \vectt_3, \vectt_4)$, $\varphi$ and $\psi$ are disjoint;
\item $\pre(\vecx)$ is in Disjunctive Normal Form (DNF);
\item $\post(\vecx)$ is in Conjunctive Normal Form (CNF).
\end{enumerate}
Obviously, any invariant synthesis problem can be converted to above normal form. 
For each conditional term $\ite(\psi, \vecx + \vecc, \vectt_2)$ occurred in the normalized $\trans(\vecx)$, we call $(\psi, \vecc)$ a \emph{branch}.

\begin{lemma}
For any translational CLIA term $\trans(\vecx)$ and any CLIA condition $\varphi(\vecx)$, there exists another CLIA term, denoted as $\trans(\varphi)$, such that $\forall \vecx \Big( \varphi(\vecx) \leftrightarrow \trans(\varphi)(\trans(\vecx)) \big)$.
\end{lemma}
\begin{proof}
Let $\trans$ be normalized and $\branches(\trans)$ be the set of branches within $\trans$. Then one can construct the formula as follows: $$\trans(\varphi)(\vecx) \equiv \displaystyle \bigvee_{(\psi, \vecc) \in \branches(\trans)} \varphi(\vecx - \vecc)$$
The formula enumerates all branches and each $\varphi(\vecx - \vecc)$ characterizes the set of values that satisfy the branch guard before transition. Notice that $\varphi(\vecx - \vecc)$ can be converted to a standard linear arithmetic formula: if $\varphi$ is $\vecc_0 \cdot \vecx \geq d$, then $\varphi(\vecx - \vecc)$ is equivalent to $\vecc_0 \cdot \vecx \geq d + \vecc_0 \cdot \vecc$; if $\varphi$ is a boolean combination of atomic formulas, then $\varphi(\vecx - \vecc)$ can be converted recursively.
\end{proof}


We further refine $\trans$ to two kinds of transitions: local transitions and non-local transitions.

\begin{lemma}
For each normalized translational term $\trans(\vecx)$, one can define two other translational terms, denoted as $\trans_l$ and $\trans_c$, such that for any $\vecx$ and $\vecy$, $\trans_l(\vecx) = \vecy$ if and only if $\trans(\vecx) = \vecy$ and $\vecx$ and $\vecy$ belong to the same branch; $\trans_c(\vecx) = \vecy$ if and only if $\trans(\vecx) = \vecy$ and $\vecx$ and $\vecy$ belong to two different branches.
\end{lemma}
\begin{proof}
The conditional term for $\trans_l(\vecx)$ can be defined by replacing every $\ite(\psi(\vecx), \vecx + \vecc, \vectt)$ in $\trans(\vecx)$ with $\ite(\psi(\vecx) \wedge \psi(\vecx + \vecc), \vecx + \vecc, \vectt)$. Similarly, $\trans_c(\vecx)$ can be defined by replacing every $\ite(\psi(\vecx), \vecx + \vecc, \vectt)$ in $\trans(\vecx)$ with $\ite(\psi(\vecx) \wedge \neg \psi(\vecx + \vecc), \vecx + \vecc, \vectt)$.
\end{proof}

We write $\trans^*$ to denote the transitive closure of $\trans$. Intuitively, $\trans^*$ represents arbitrary sequences of transitions of $\trans$, i.e., $\trans_l$ or $\trans_c$. Then the sequence of transitions can be divided into multiple segments using cross-branch transitions such that each segment consists of same-branch transitions only. This intuition can be formally captured by the following lemma:
\begin{lemma}
For any transition $\trans$ and any constraint $\varphi$, $\trans^*(\varphi)$ if and only if $\trans_l^* \circ (\trans_c \circ \trans_l^*)^*(\varphi)$.
\end{lemma}

\subsection{The decidable fragment}

If we start from the precondition $\pre$ and compute the transitive closure, $\trans^*(\pre)$ is just the strongest loop invariant, if any invariant exists ($\trans^*(\pre) \rightarrow \post$). However in general, $\trans^*(\pre)$ is not computable as it may take arbitrarily many steps of transitions. In this paper, we show that the local transition $\trans_l^*$ is \emph{computable} and our decidable fragment just restricts $\trans$ such that the number of $\trans_c$ transitions is \emph{finite}.

We first show $\trans_l^*$ can be captured by an extended CLIA condition.

\begin{definition}
An extended CLIA condition is a boolean combination of regular atomic conditions $t_1 \geq t_2$ and modulo assertions of the form $x \mod c = d$. 
\end{definition}

\begin{lemma}\label{fast}
For any translational term $\trans(\vecx)$ and any CLIA condition $\varphi$, one can compute an extended CLIA condition, denoted as $\fast(\varphi)$, such that $\fast(\varphi) \equiv \trans_l^*(\varphi)$.
\end{lemma}
\begin{proof}
Notice that $\trans_l^*$ preserves the branch, i.e., for any $\vecx$, $\vecx$ and $\trans_l^*(\vecx)$ belong to the same branch. Hence $$\trans_l^*(\varphi) \equiv \displaystyle \bigvee_{(\psi, \vecc) \in \branches(\trans)} \trans_l^*(\varphi \wedge \psi)$$ Now it suffices to construct a formula equivalent to $\trans_l^*\big(\varphi \wedge \psi \big)$ for arbitrary branch $(\psi, \vecc)$.
Remember $(\psi, \vecc)$ corresponds to a transition $\ite(\psi(\vecx), \vecx + \vecc, t)$, and by the definition of $\trans_l$, the above transitive closure is equivalent to $$\displaystyle \bigvee_{k \geq 0} \bigg( \varphi(\vecx - k \cdot \vecc) \wedge \Big( \bigwedge_{0 \leq i \leq k} \psi(\vecx - i \cdot \vecc) \Big) \bigg)$$ As $\trans$ is normalized and every condition $\psi$ is a conjunction of atomic formulas, therefore the subformula $\displaystyle \bigwedge_{0 \leq i \leq k} \psi(\vecx - i \cdot \vecc)$ is equivalent to the conjunction of $\displaystyle \bigwedge_{0 \leq i \leq k} \vecc_j \cdot (\vecx - i \cdot \vecc) \geq d_j$ for all atomic formula of form $\vecc_j \cdot \vecx \geq d_j$ in $\psi$. As $\vecc_j \cdot (\vecx - i \cdot \vecc)$ monotonically increases or decreases along with $i$, the formula can be evaluated with the case $i=0$ or $i=k$ only: $\vecc_j \cdot (\vecx - k \cdot \vecc) \geq d_j \wedge \vecc_j \cdot \vecx \geq d_j$. Note that the first conjunct is a linear constraint over $\vecx - k \cdot \vecc$ and the second conjunct does not involve $k$; so what remains is to construct a formula equivalent to $$\displaystyle \bigvee_{k \geq 0} \varphi(\vecx - k \cdot \vecc) \wedge \psi(\vecx - k \cdot \vecc)$$ for arbitrary $\varphi$, $k$ and $\vecc$.

Notice that the above formula is semantically equivalent to a quantified formula in LIA: $$\exists k \Big( k \geq 0 \wedge \varphi(\vecx - k \cdot \vecc) \wedge \psi(\vecx - k \cdot \vecc) \Big)$$ which by quantifier elimination, can be further reduced to an equivalent, quantifier-free formula $\gamma(\vecx)$. Notice that the $\gamma(\vecx)$ is an extended CLIA condition but not necessarily a regular CLIA condition: it may contain modulo assertions of the form $x \mod C = d$, where $x \in \vecx$ and $C$ can be divided by any constant from $\vecc$.

\end{proof}

Now to define the decidable fragment, we need to capture the intuition ``the number of $\trans_c$ is finite'' by defining the transition graph and its acyclicity.

\begin{definition}
Given a normalized translational term $\trans(\vecx)$, the transition graph with respect to the transition is a tuple $(V, E)$ where
\begin{itemize}
\item $V = \{ v_{\psi, \vecc} \mid (\psi, \vecc) \in \branches(\trans(\vecx)) \}$ is the set of vertices.
\item $E = \{ (v_{\psi, \vecc}, v_{\psi', \vecc'}) \mid \psi(\vecx) \wedge \psi'(\vecx+\vecc) \textrm{ is satisfiable} \}$ is the set of directed edges.
\end{itemize}
\end{definition}
\begin{definition}
An invariant synthesis problem $\exists \inv \forall \vecx \varphi(\inv; \vecx)$ is Acyclic Translational (AT) if: a) it is translational; and b) its transition graph is acyclic.
\end{definition}

\begin{theorem}[Decidability of AT]
There is a terminating synthesis procedure that takes as input an AT invariant synthesis problem, and either outputs a loop invariant or concludes the nonexistence of a CLIA invariant.
\end{theorem}
\begin{proof}
As mentioned above, the strongest loop invariant can be computed as $\trans^*(\pre) \equiv \displaystyle \bigvee_{k\geq0} \trans_l^* \circ (\trans_c \circ \trans_l^*)^k(\pre)$. By Lemma~\ref{fast}, $\trans_l^*$ can be expressed as $\fast$; and by the acyclic assumption, the number of $\trans_c$ is bounded by the diameter (length of the longest path) of the transition graph, say $n$. Then the procedure can simply generate the following invariant and check its validity: $$\textit{strong-inv}(\vecx) \equiv \displaystyle \bigvee_{k=1}^{n} \fast \circ (\trans_c \circ \fast)^k(\pre)$$

Notice that $\textit{strong-inv}(\vecx)$ is not necessarily a CLIA invariant. In the case that $\textit{strong-inv}$ involves modulo assertions, we simply replace them with $\top$ and obtain a weaker condition $\textit{sub-strong-inv}(\vecx)$. We argue that if $\textit{strong-inv}(\vecx)$ is an invariant, so is $\textit{sub-strong-inv}(\vecx)$. Intuitively, the evaluation of a modulo assertion $x \mod C$ does not affect the evaluation of any non-modulo assertion in $\textit{strong-inv}(\vecx)$.
\end{proof}

\subsubsection{Extension for Constant Varaibles}

The decidability proof for AT invariant synthesis problems can be naturally extended to handle constant variables. Formally, if the only occurrence of a variable $x_j$ in $\trans(\vecx)$ is just $x_j$, then $x_j$ is called a constant variable and can be treated as a constant. In other words, the translational transition $\trans(\vecx)$ may allow $x_i + x_j$ for arbitrary $x_i \in \vecx$. This class of invariant synthesis problem can be decided using the same algorithm set forth above.

\section{Experimental Evaluation}

We have prototyped our concolic synthesis framework as \name, which has embodied all three synthesis algorithms presented in this paper, namely concolic search, SSI synthesis and AT invariant synthesis. The concolic search algorithm presented in Section~\ref{sec:concolic} is the default synthesis engine that supports arbitrary CLIA synthesis problems. When the input synthesis problem falls into the SSI fragment or AT fragment, the corresponding decidable synthesis procedure supersedes concolic search (unless the user turns this option off).

\begin{figure}[htp!]
\begin{center}
\includegraphics[width=.8\columnwidth]{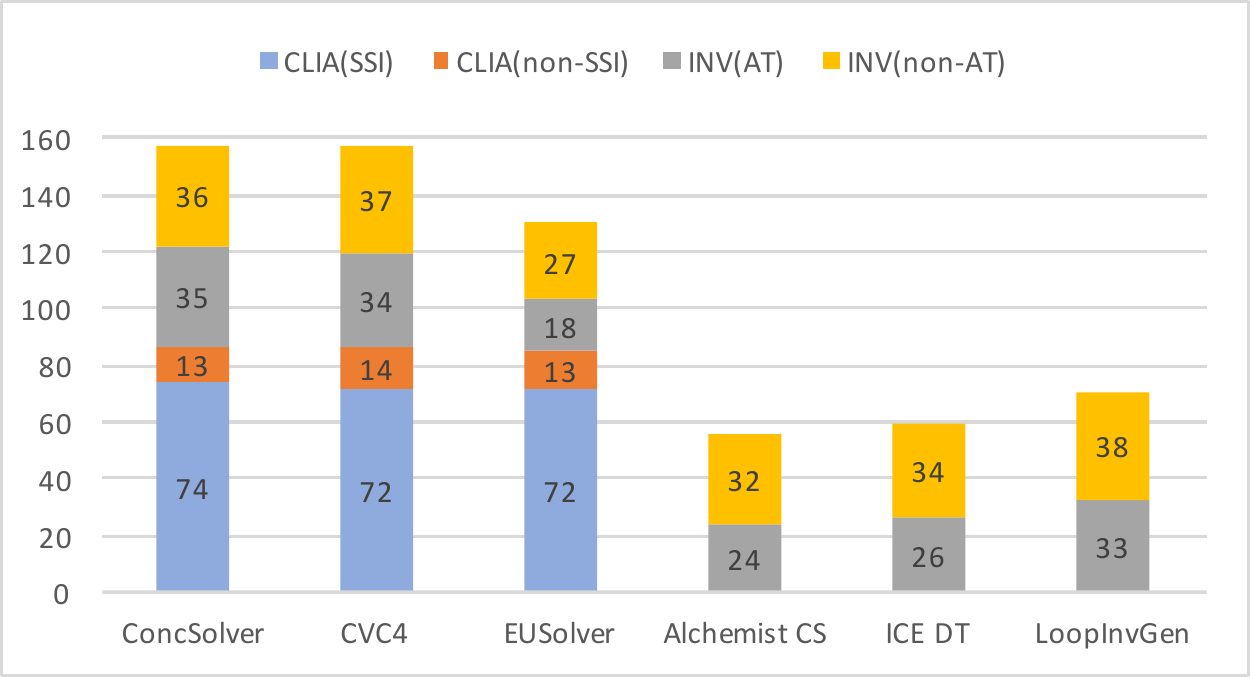}
\end{center}
\caption{Solved benchmarks (breakdown by categories)}
\label{fig:total}
\end{figure}

\begin{figure}[htp!]
\begin{center}
\includegraphics[width=.8\columnwidth]{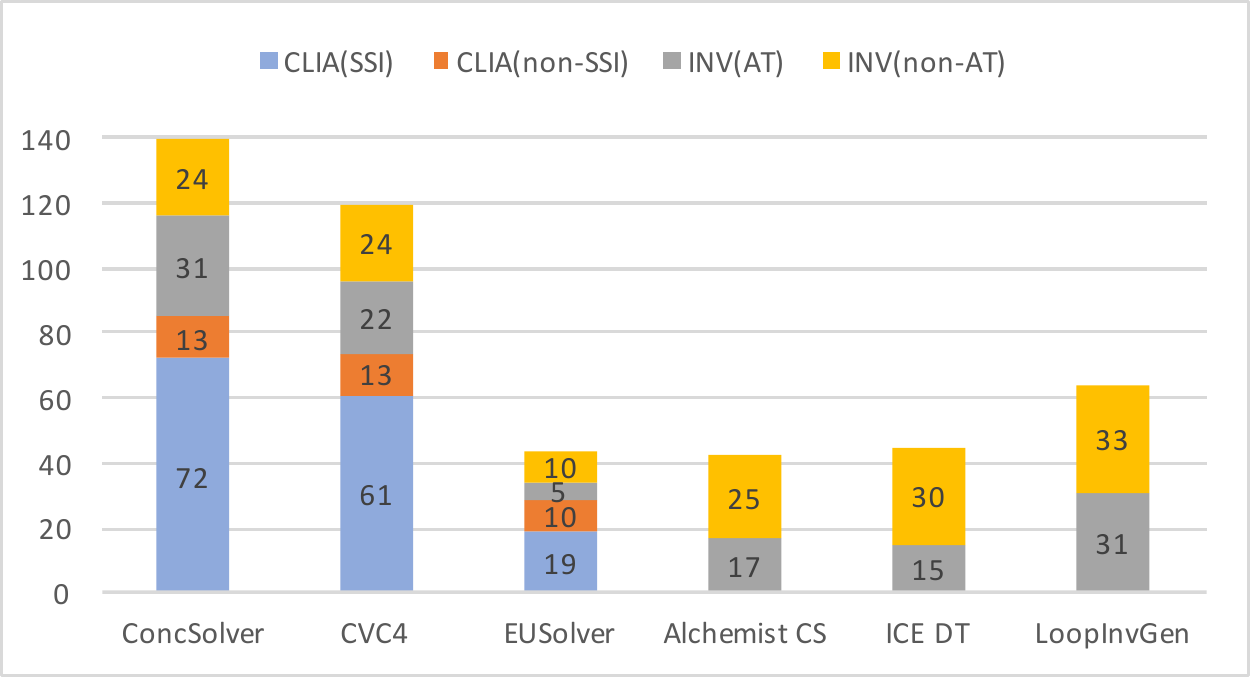}
\end{center}
\caption{Fastest solved benchmarks (breakdown by categories)}
\label{fig:fastest}
\end{figure}

%

\begin{figure}[!htb]
\centering
\small
\begin{subfigure}[b]{.45\textwidth}
\centering
\includegraphics[width=\columnwidth]{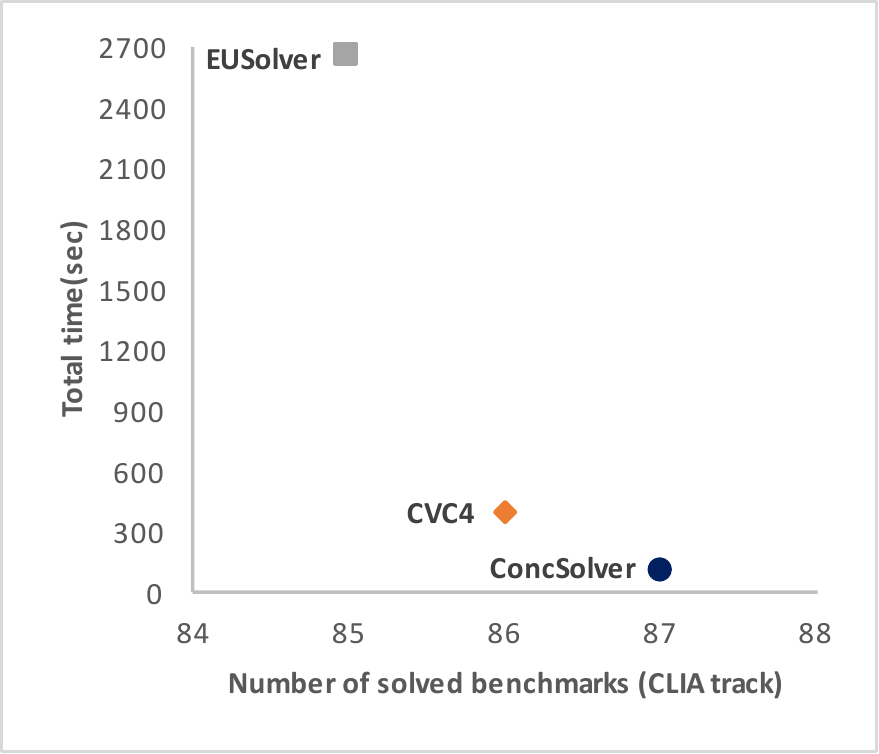}
\caption{CLIA track}
\label{fig:time-clia}
\end{subfigure}
\hfill
\begin{subfigure}[b]{.45\textwidth}
\centering
\includegraphics[width=\columnwidth]{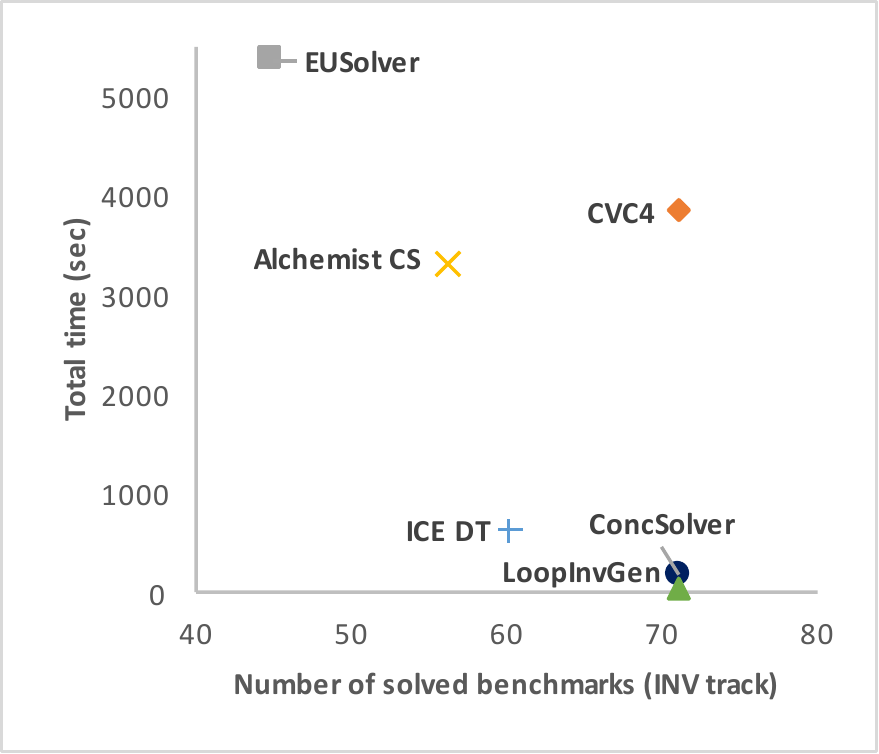}
\caption{INV track}
\label{fig:time-inv}
\end{subfigure}
\caption{Total solving time vs. \# solved benchmarks}
\label{fig:time}
\end{figure}

%

\begin{figure}[!htb]
\centering
\small
\begin{subfigure}[b]{.45\textwidth}
\centering
\includegraphics[width=\columnwidth]{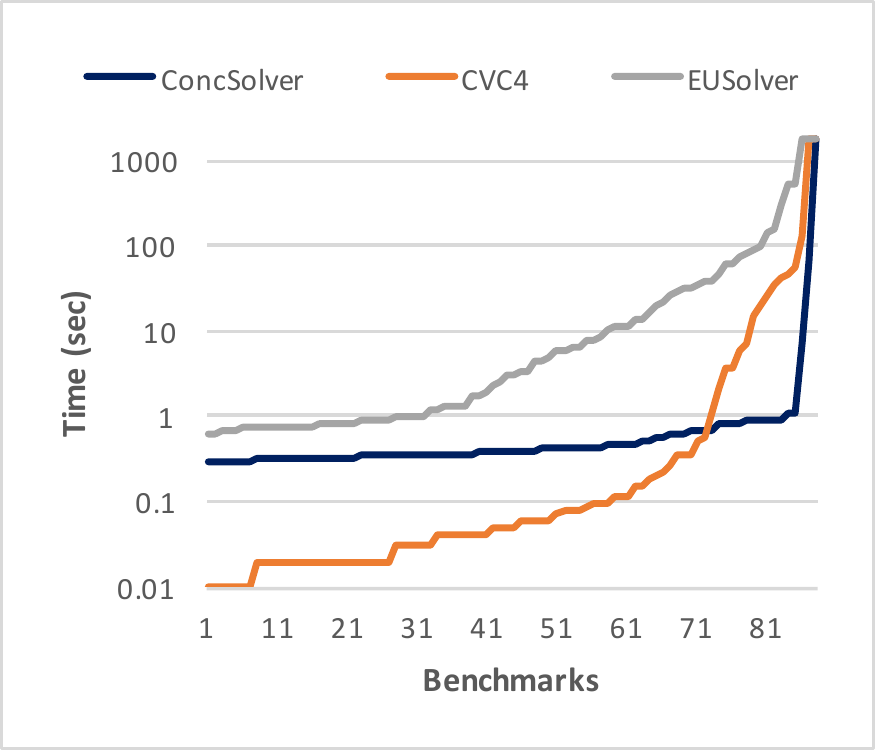}
\caption{CLIA track}
\label{fig:acc-clia}
\end{subfigure}
\hfill
\begin{subfigure}[b]{.45\textwidth}
\centering
\includegraphics[width=\columnwidth]{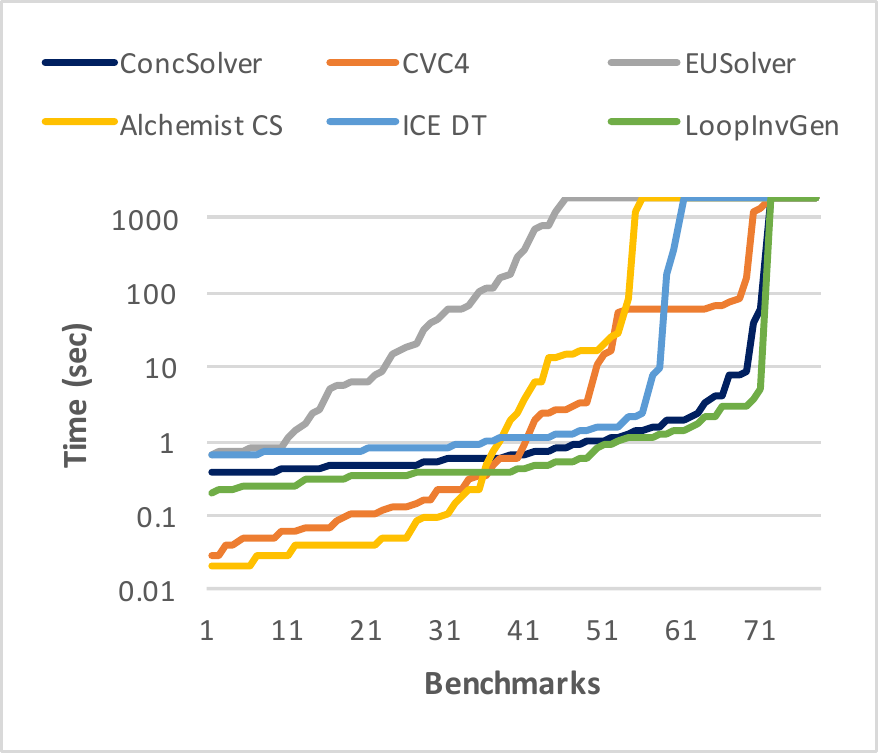}
\caption{INV track}
\label{fig:acc-inv}
\end{subfigure}
\caption{Solving time per benchmark}
\label{fig:acc}
\end{figure}

\begin{figure}
\centering
\footnotesize
\begin{tabular}{|c|c|c|r|}
\hline
Benchmark & Category & Solved by & Time (sec) \\
\hline
\textsf{large\_linear\_func} & \textsf{CLIA(SSI)} & \name & 0.3 \\
\hline
\textsf{small\_linear\_func} & \textsf{CLIA(SSI)} & \name & 0.3 \\
\hline
\textsf{ex\_23\_vars} & \textsf{INV(AT)} & \name & 0.47 \\
\hline
\textsf{hola.07} & \textsf{INV(non-AT)} & \textsc{LoopInvGen} & 2.24 \\
\hline
\end{tabular}
\caption{Benchmarks solved uniquely by a solver}
\label{fig:unique}
\end{figure}


%
%

To evaluate our algorithms, we compared \name with several participant solvers of past two years' SyGuS competition in the CLIA and INV tracks: \textsc{Alchemist-CS}~\cite{alchemist}, CVC4~\cite{Reynolds2015}, \textsc{EUSolver}~\cite{Alur2017}, ICE-DT~\cite{ice-dt} and \textsc{LoopInvGen}~\cite{Padhi2017}. CVC4 won the CLIA track in both 2016 and 2017; ICE-DT and \textsc{LoopInvGen} won the INV track in 2016 and 2017, respectively~\cite{sygus16,sygus17}.
 
We adopted all 73 benchmarks from the CLIA track (out of which 59 fall in the SSI fragment) and all 74 benchmarks from the INV track (out of which 31 fall in the AT fragment). To make the set of benchmarks more diverse, we manually wrote 5 benchmarks with the SSI property as well as commutativity, and translated 4 benchmarks to the SyGuS interchange format from SV-COMP 2017~\cite{svcomp2017} and some programs of our own. We also extend the CLIA benchmarks with 10 larger-sized examples (e.g., extending \textsf{max15} to \textsf{max20}). The extensions wound up with a set of 166 benchmarks~\footnote{available at the anonymized link: \url{https://www.dropbox.com/s/wzvw4jxfq7yiz4z/benchmarks.zip?dl=0}}. Experiments were conducted on the StarExec platform~\cite{starexec}, on which each solver is executed on a 4-core, 2.4GHz CPU and 128GB memory node, with a 30-minute timeout.

We summarize the experimental results through a set of figures (see Appendix~\ref{sec:tables} for more details). Figure~\ref{fig:total} compares the number of benchmarks correctly solved within the 30-minute limit. Figure~\ref{fig:fastest} compares the number of benchmarks solved the fastest among all solvers~\footnote{Following the criterion of SyGuS competition, the time amounts are classified into buckets of pseudo-logarithmetic scales: $[0,1), [1,3), [3,10), \dots, [1000,1800)$.}. Both figures break down the numbers into four categories: \textsf{CLIA(SSI)} and \textsf{INV(AT)} represent benchmarks belonging to our SSI and AT fragments, respectively; \textsf{CLIA(non-SSI)} and \textsf{INV(non-AT)} represent the remaining CLIA and INV synthesis benchmarks, respectively. Figure~\ref{fig:time} shows the relative effectiveness and efficiency of the solvers in terms of the number of benchmarks solved and the total amount of time spent. Figure~\ref{fig:acc} shows the amounts of time spent by various benchmarks, sorted in ascending order. Figure~\ref{fig:unique} presents the benchmarks that were solved by one of the solvers uniquely.

Observing the figures, we are encouraged by the following facts: 1) overall, \name solved and solved fastest more benchmarks than any other solver; 2) our SSI synthesis algorithm solved \emph{all 74 benchmarks} in the category, among which 72 were solved fastest and 2 were solved uniquely; 3) our AT invariant synthesis algorithm solved \emph{all 35 benchmarks} in the category, among which 31 were solved fastest and 1 were solved uniquely; 4) \name solved one more CLIA benchmarks than the winner CVC4, and the same number of INV benchmarks as the winner {\sc LoopInvGen}, both with comparable performance.

In summary, our concolic synthesis framework tends to be an effective engine for a wide variety of synthesis problems, and holds promise of performing synthesis at a larger scale than possible before.

\section{Related Work}

\subsubsection{Syntax-guided synthesis.}

As we mentioned in Section~\ref{sec:intro}, the most important dimension along which we characterize existing syntax-guided synthesis approaches is their search strategies. For example, among state-of-the-art solvers participating recent years' SyGuS competition, EUSolver~\cite{Alur2017} adopts a purely enumerative search strategy, and CVC4~\cite{Reynolds2015} solves the synthesis problem purely symbolically through SMT solving. The Counterexample-Guided Inductive Synthesis (CEGIS) framework~\cite{program-sketching} has been a common theme underlying several solvers which differ in how the synthesizer generates candidates from counterexamples: Sketch~\cite{sketch:manual,sketchthesis} solves constraints encoding the counterexamples; Alchemist~\cite{alchemist} and ICE-DT~\cite{ice-dt} find likely candidates using learning algorithms. 

It is noteworthy that Sketch-AC~\cite{cav15,fmsd17} is an earlier attempt to combine enumerative and symbolic search, but in a very different way. As a Sketch-based tool, Sketch-AC supports a more general class of SyGuS problems than CLIA synthesis. Its algorithm statistically determines a class of \emph{highly influential unknowns} and explicitly enumerates all possible values of these unknowns. Unlike our decision-tree-based enumeration, their enumeration strategy is not specialized for CLIA synthesis and seems not very competitive in synthesizing CLIA functions~\cite{sygus16}.

The decision tree representation of functions to be synthesized has been used in the ICE-DT learning algorithm~\cite{ice-dt} and the enumerative search algorithm underlying EUSolver~\cite{Alur2017}. Both the two algorithms use decision trees to represent functions in order to leverage standard decision-tree-based learning algorithms to generate the whole function or some components. Our concolic synthesis framework uses decision trees in a different way: the tree representation encodes the task of synthesizing functions to the task of synthesizing coefficients, and the tree height is a natural measurement for enumeration.

A lot of syntax-guided synthesis algorithms have been developed in the past few years, including learning-based ones~\cite{alchemist,ice-cs,ice-dt} and enumerative ones~\cite{TRANSIT,unification}. However, they did not present any theoretical results such as termination, decidability, or complexity.
Caulfield~\etal~\cite{Caulfield} identified several decidable fragments of synthesis problems in the theories of EUF and bit-vectors, but they did not report any decidable or undecidable results for the theory of CLIA.

\subsubsection{Invariant synthesis.}

There is a rich literature of invariant synthesis and we only discuss the work mostly related to ours. Our AT invariant synthesis algorithm is distinctive from most of existing invariant synthesis approaches in the sense that our algorithm is terminating and complete, without assuming any syntactic template of the invariant. Constraint-based and abstract-interpretation-based approaches typically presume a shape of the invariant to be synthesized, e.g., octagon, polyhedra, conjunction-only constraints, etc. In contrast, CEGAR-based and interpolation-based approaches can synthesize arbitrary CLIA invariants but usually do not guarantee the completeness. Several recently developed invariant synthesis algorithms such as Hola~\cite{hola} or FiB~\cite{fib} bear similarities to ours: they all start from the precondition (postcondition) and repeatedly weakening (strengthening) the invariant until it becomes inductive. Our procedure is different from these algorithms as the translational transition allows us to accelerate the transformation by quantifier elimination, and guarantees the termination of the algorithm. 

In recent years, the CEGIS framework has been adopted by several invariant synthesis algorithms~\cite{ice-cs,ice-dt,Sharma2013,Li2017,Padhi2017}. All these algorithms leverage learning techniques in their backend synthesizer.

\subsubsection{Single Invocation Property.}

Our synthesis procedure for the SSI fragment is very relevant to previous procedures~\cite{Reynolds2015,Kuncak2010,Kuncak2012}. The SSI property strengthens the Single Invocation (SI) property proposed by Reynolds~\etal~\cite{Reynolds2015}. The main difference between SSI and SI is that the latter is more general and does not require the single occurrence of $f$ in each atomic formula. As a result, their CEGQI algorithm heuristically instantiates the singly invoked function with candidate terms and may not terminate. The software synthesis procedures developed by Kuncak~\etal~\cite{Kuncak2010,Kuncak2012} essentially handles CLIA synthesis problems with the SI property as well, but guarantees termination and completeness. Our procedure is also terminating and complete, and differs from theirs in two aspects. First, their procedure is more general and essentially solves CLIA synthesis problem with SI properties. The cost is DNF conversion and quantifier elimination, both leading to exponential blowups. Our approach is more restrictive and focuses on the SSI fragment only, leading to a more efficient polynomial-time procedure. Second, their quantifier-elimination-based procedure yields solutions involving floor and ceiling functions, which are not allowed in our CLIA synthesis problem.

\bibliographystyle{splncs03}

\bibliography{refs}

\appendix
\section{Experimental Results}
\label{sec:tables}
\begin{table}
\fontsize{8}{10}\selectfont
\begin{tabu} to \textwidth {|X[c,12]|X[r,7]|X[r,4]|X[r,6]|X[r,8]|X[r,5]|X[r,7]|}
\hline
Benchmark & ConcSolver & CVC4 & EUSolver & Alchemist CS & ICE DT & LoopInvGen \\
\hline
\sf{fg\_max2}& 0.3 & \textbf{0.01} & 1.35 & - & - & - \\ 
\sf{fg\_max3}& 0.29 & \textbf{0.02} & 0.83 & - & - & - \\ 
\sf{fg\_max4}& 0.32 & \textbf{0.02} & 1.01 & - & - & - \\ 
\sf{fg\_max5}& 0.48 & \textbf{0.04} & 0.96 & - & - & - \\ 
\sf{fg\_max6}& 0.33 & \textbf{0.06} & 3.06 & - & - & - \\ 
\sf{fg\_max7}& 0.34 & \textbf{0.12} & 2.57 & - & - & - \\ 
\sf{fg\_max8}& 0.36 & \textbf{0.2} & 1.95 & - & - & - \\ 
\sf{fg\_max9}& \textbf{0.35} & 0.37 & 3.31 & - & - & - \\ 
\sf{fg\_max10}& \textbf{0.37} & 0.57 & 4.64 & - & - & - \\ 
\sf{fg\_max11}& \textbf{0.37} & 1.05 & 9.86 & - & - & - \\ 
\sf{fg\_max12}& \textbf{0.38} & 2.13 & 13.75 & - & - & - \\ 
\sf{fg\_max13}& \textbf{0.38} & 3.79 & 14.22 & - & - & - \\ 
\sf{fg\_max14}& \textbf{0.39 }& 5.85 & 26.8 & - & - & - \\ 
\sf{fg\_max15}& \textbf{0.4} & 7.17 & 20.48 & - & - & - \\ 
\sf{fg\_max16}& \textbf{0.41} & 15.15 & 60.67 & - & - & - \\ 
\sf{fg\_max17}& \textbf{0.41} & 27.49 & 90.33 & - & - & - \\ 
\sf{fg\_max18}& \textbf{0.43} & 44.72 & 152.09 & - & - & - \\ 
\sf{fg\_max19}& \textbf{0.44} & 56.6 & 508.88 & - & - & - \\ 
\sf{fg\_max20}& \textbf{0.45} & 134.89 & 541.19 & - & - & - \\ 
\sf{fg\_array\_search\_2}& 0.34 & \textbf{0.01} & 0.75 & - & - & - \\ 
\sf{fg\_array\_search\_3}& 0.7 & \textbf{0.02} & 0.85 & - & - & - \\ 
\sf{fg\_array\_search\_4}& 0.34 & \textbf{0.02} & 1.03 & - & - & - \\ 
\sf{fg\_array\_search\_5}& 0.37 & \textbf{0.02} & 1.28 & - & - & - \\ 
\sf{fg\_array\_search\_6}& 0.4 & \textbf{0.04} & 1.68 & - & - & - \\ 
\sf{fg\_array\_search\_7}& 0.4 & \textbf{0.04} & 2.27 & - & - & - \\ 
\sf{fg\_array\_search\_8}& 0.44 & \textbf{0.04} & 3.21 & - & - & - \\ 
\sf{fg\_array\_search\_9}& 0.46 & \textbf{0.04} & 4.52 & - & - & - \\ 
\sf{fg\_array\_search\_10}& 0.52 & \textbf{0.05} & 6.35 & - & - & - \\ 
\sf{fg\_array\_search\_11}& 0.56 & \textbf{0.06} & 8.6 & - & - & - \\ 
\sf{fg\_array\_search\_12}& 0.65 & \textbf{0.07} & 11.78 & - & - & - \\ 
\sf{fg\_array\_search\_13}& 0.66 & \textbf{0.08} & 16.01 & - & - & - \\ 
\sf{fg\_array\_search\_14}& 0.87 & \textbf{0.09} & 21.48 & - & - & - \\ 
\sf{fg\_array\_search\_15}& 0.9 & \textbf{0.1} & 28.44 & - & - & - \\ 
\sf{fg\_array\_search\_16}& 0.83 & \textbf{0.11} & 36.05 & - & - & - \\ 
\sf{fg\_array\_search\_17}& 0.8 & \textbf{0.12} & 47.51 & - & - & - \\ 
\sf{fg\_array\_search\_18}& 0.93 & \textbf{0.15} & 59.91 & - & - & - \\ 
\sf{fg\_array\_search\_19}& 1.08 & \textbf{0.15} & 75.58 & - & - & - \\ 
\sf{fg\_array\_search\_20}& 1.08 & \textbf{0.18} & 94.57 & - & - & - \\
\hline
\end{tabu}
\caption{CLIA(SSI) Benchmarks Part 1}
\label{tbl:ssi_benchmark_1}
\end{table}

\begin{table}
\fontsize{8}{10}\selectfont
\begin{tabu} to \textwidth {|X[c,12]|X[r,7]|X[r,4]|X[r,6]|X[r,8]|X[r,5]|X[r,7]|}
\hline
Benchmark & ConcSolver & CVC4 & EUSolver & Alchemist CS & ICE DT & LoopInvGen \\
\hline
\sf{fg\_array\_sum\_2\_15}& 0.31 & \textbf{0.01} & 0.71 & - & - & - \\ 
\sf{fg\_array\_sum\_2\_5}& 0.29 & \textbf{0.01} & 0.7 & - & - & - \\ 
\sf{fg\_array\_sum\_3\_15}& 0.31 & \textbf{0.02} & 0.76 & - & - & - \\ 
\sf{fg\_array\_sum\_3\_5}& 0.31 & \textbf{0.02} & 0.81 & - & - & - \\ 
\sf{fg\_array\_sum\_4\_15}& 0.33 & \textbf{0.03} & 0.87 & - & - & - \\ 
\sf{fg\_array\_sum\_4\_5}& 0.34 & \textbf{0.03} & 0.87 & - & - & - \\ 
\sf{fg\_array\_sum\_5\_15}& 0.34 & \textbf{0.03} & 1.36 & - & - & - \\ 
\sf{fg\_array\_sum\_6\_15}& 0.37 & \textbf{0.04} & 5.92 & - & - & - \\ 
\sf{fg\_array\_sum\_6\_5}& 0.35 & \textbf{0.04} & 6.05 & - & - & - \\ 
\sf{fg\_array\_sum\_7\_15}& 0.39 & \textbf{0.05} & 7.75 & - & - & - \\ 
\sf{fg\_array\_sum\_7\_5}& 0.39 & \textbf{0.05} & 7.64 & - & - & - \\ 
\sf{fg\_array\_sum\_8\_15}& 0.41 & \textbf{0.06} & 11 & - & - & - \\ 
\sf{fg\_array\_sum\_8\_5}& 0.41 & \textbf{0.06} & 11.24 & - & - & - \\ 
\sf{fg\_array\_sum\_9\_15}& 0.42 & \textbf{0.08} & 31.73 & - & - & - \\ 
\sf{fg\_array\_sum\_9\_5}& 0.41 & \textbf{0.08} & 30.93 & - & - & - \\ 
\sf{fg\_array\_sum\_10\_5}& 0.45 & \textbf{0.1} & 39.18 & - & - & - \\ 
\sf{fg\_array\_sum\_10\_15}& 0.45 & \textbf{0.1} & 38.65 & - & - & - \\ 
\sf{fg\_mpg\_example1}& 0.32 & \textbf{0.02} & 0.83 & - & - & - \\ 
\sf{fg\_mpg\_example2}& 0.35 & \textbf{0.02} & 1.79 & - & - & - \\ 
\sf{fg\_mpg\_example3}& 0.36 & \textbf{0.04} & 3.06 & - & - & - \\ 
\sf{fg\_mpg\_example4}& 0.36 & \textbf{0.05} & 6.46 & - & - & - \\ 
\sf{fg\_mpg\_example5}& 0.39 & \textbf{0.03} & 4.66 & - & - & - \\ 
\sf{fg\_mpg\_guard1}& 0.31 & \textbf{0.01} & 0.81 & - & - & - \\ 
\sf{fg\_mpg\_guard2}& 0.32 & \textbf{0.02} & 0.87 & - & - & - \\ 
\sf{fg\_mpg\_guard3}& 0.32 & \textbf{0.02} & 0.87 & - & - & - \\ 
\sf{fg\_mpg\_guard4}& 0.32 & \textbf{0.02} & 0.88 & - & - & - \\ 
\sf{fg\_mpg\_ite1}& 0.33 & \textbf{0.02} & 0.96 & - & - & - \\ 
\sf{fg\_mpg\_ite2}& 0.34 & \textbf{0.02} & 1.3 & - & - & - \\ 
\sf{fg\_mpg\_plane1}& 0.32 & \textbf{0.02} & 0.76 & - & - & - \\ 
\sf{fg\_mpg\_plane2}& 0.33 & \textbf{0.01} & 0.77 & - & - & - \\ 
\sf{fg\_mpg\_plane3}& 0.33 & \textbf{0.01} & 0.77 & - & - & - \\ 
\sf{diff}& \textbf{0.34} & 20.19 & 144.38 & - & - & - \\ 
\sf{large\_linear\_func}& \textbf{0.3} & $\infty$ & $\infty$ & - & - & - \\ 
\sf{large}& \textbf{0.29} & 0.51 & 1.17 & - & - & - \\ 
\sf{small\_linear\_func}& \textbf{0.3} & $\infty$ & $\infty$ & - & - & - \\ 
\sf{small}& 0.3 & \textbf{0.28} & 1.19 & - & - & - \\ 
\hline
\end{tabu}
\caption{CLIA(SSI) Benchmarks Part 2}
\label{tbl:ssi_benchmark_2}
\end{table}

\begin{table}
\fontsize{8}{10}\selectfont
\begin{tabu} to \textwidth {|X[c,12]|X[r,7]|X[r,4]|X[r,6]|X[r,8]|X[r,5]|X[r,7]|}
\hline
Benchmark & ConcSolver & CVC4 & EUSolver & Alchemist CS & ICE DT & LoopInvGen \\
\hline
\sf{fg\_fivefuncs} & 0.67 & \textbf{0.02} & 0.73 & - & - & - \\ 
\sf{fg\_sixfuncs} & 0.71 & \textbf{0.03} & 0.73 & - & - & - \\ 
\sf{fg\_sevenfuncs} & 0.79 & \textbf{0.02} & 0.74 & - & - & - \\ 
\sf{fg\_eightfuncs} & 0.83 & \textbf{0.02} & 0.75 & - & - & - \\ 
\sf{fg\_ninefuncs} & 0.92 & \textbf{0.03} & 0.76 & - & - & - \\ 
\sf{fg\_tenfunc1} & 0.88 & \textbf{0.35} & 0.78 & - & - & - \\ 
\sf{fg\_tenfunc2} & 0.94 & \textbf{0.35} & 0.79 & - & - & - \\ 
\sf{fg\_polynomial} & $\infty$ & \textbf{0.02} & 0.65 & - & - & - \\ 
\sf{fg\_polynomial1} & 0.57 & \textbf{0.02} & 0.65 & - & - & - \\ 
\sf{fg\_polynomial2} & 0.61 & \textbf{0.06} & 0.71 & - & - & - \\ 
\sf{fg\_polynomial3} & 0.51 & \textbf{0.23} & 1.01 & - & - & - \\ 
\sf{fg\_polynomial4} & \textbf{0.65} & 40.94 & 311.78 & - & - & - \\ 
\sf{fg\_VC22\_a} & 7.23 & \textbf{3.57} & 79 & - & - & - \\ 
\sf{fg\_VC22\_b} & 64.9 & \textbf{34.04} & $\infty$ & - & - & - \\ 
\hline
\end{tabu}
\caption{CLIA(non-SSI) Benchmarks}
\label{tbl:clia_benchmark}
\end{table}

\begin{table}
\vspace{-.22in}
\fontsize{8}{10}\selectfont
\begin{tabu} to \textwidth {|X[c,12]|X[r,7]|X[r,4]|X[r,6]|X[r,8]|X[r,5]|X[r,7]|}
\hline
Benchmark & ConcSolver & CVC4 & EUSolver & Alchemist CS & ICE DT & LoopInvGen \\
\hline
\sf{cegar1\_vars-new} & 7.35 & \textbf{0.35} & $\infty$ & 1.01 & 0.83 & 0.4 \\ 
\sf{cegar1\_vars} & 0.59 & 0.34 & 63.5 & \textbf{0.03} & 1.16 & 0.38 \\ 
\sf{cegar1-new} & 7.35 & 0.16 & $\infty$ & \textbf{0.09} & 0.76 & 0.83 \\ 
\sf{cegar1} & 0.6 & 0.16 & 8.04 & \textbf{0.04} & 0.71 & 0.4 \\ 
\sf{cggmp-new} & 0.46 & \textbf{0.05} & $\infty$ & $\infty$ & 391.96 & 5 \\ 
\sf{cggmp} & 0.5 & \textbf{0.04} & error & 82.92 & 1.55 & 3.62 \\ 
\sf{dec\_simpl-new\footnotemark}  & 0.41 & \textbf{0.05} & $\infty$ & $\infty$ & $\infty$ & 0.39 \\ 
\sf{dec\_simpl} & 0.4 & 0.08 & 2.45 & \textbf{0.04} & 0.77 & 0.34 \\ 
\sf{dec\_vars-new\footnotemark[\value{footnote}]}  & 0.4 & \textbf{0.05} & $\infty$ & $\infty$ & $\infty$ & 0.35 \\ 
\sf{dec\_vars} & 0.41 & 0.23 & 16.77 & \textbf{0.05} & 0.82 & 0.34 \\ 
\sf{dec-new} & 0.4 & 0.07 & 0.7 & \textbf{0.03} & 0.64 & 0.25 \\ 
\sf{dec} & 0.4 & 0.07 & 0.7 & \textbf{0.03} & 0.63 & 0.19 \\ 
\sf{ex\_23\_vars} & \textbf{0.47} & error & $\infty$ & $\infty$ & $\infty$ & $\infty$ \\ 
\sf{ex23} & \textbf{0.48} & 1199 & $\infty$ & $\infty$ & 1.42 & $\infty$ \\ 
\sf{ex7\_vars} & 0.45 & 10.69 & 782.16 & \textbf{0.04} & 0.78 & 0.9 \\ 
\sf{ex7} & 0.44 & 2.99 & 17.46 & \textbf{0.04} & 0.73 & 0.94 \\ 
\sf{hola\_add} & 0.48 & 0.47 & 1.8 & \textbf{0.04} & $\infty$ & 0.44 \\ 
\sf{hola\_countud} & 0.47 & 60.4 & 287.59 & 3.65 & $\infty$ & \textbf{0.44} \\ 
\sf{inc\_simp} & 0.42 & 0.05 & 4.96 & \textbf{0.04} & 0.76 & 0.36 \\ 
\sf{inc\_vars} & 0.4 & 0.1 & 30.21 & \textbf{0.04} & 0.78 & 0.34 \\ 
\sf{inc} & 0.39 & \textbf{0.07} & $\infty$ & 1272.11 & 2.35 & 0.24 \\ 
\sf{sum1\_vars} & \textbf{0.47} & 60.84 & $\infty$ & $\infty$ & 1.48 & 0.5 \\ 
\sf{sum1} & \textbf{0.47} & 60.42 & 783.89 & 17.12 & 1 & 0.49 \\ 
\sf{sum3\_vars} & 0.39 & \textbf{0.03} & 0.69 & 0.05 & 0.79 & 0.25 \\ 
\sf{sum3} & 0.41 & \textbf{0.05} & 0.68 & \textbf{0.05} & 0.78 & 0.24 \\ 
\sf{sum4\_simp} & 0.46 & 60.39 & $\infty$ & 14.3 & 0.99 & \textbf{0.45} \\ 
\sf{sum4\_vars} & 0.47 & 60.83 & $\infty$ & $\infty$ & 1.15 & \textbf{0.45} \\ 
\sf{sum4} & 0.45 & \textbf{0.04} & $\infty$ & 6.38 & 2.02 & 0.21 \\ 
\sf{tacas\_vars} & 0.57 & 166.3 & 116.4 & \textbf{0.1} & 1.26 & 0.54 \\ 
\sf{tacas} & 0.55 & 16.74 & 116.95 & \textbf{0.45} & 1.15 & 0.57 \\ 
\sf{w1} & 0.43 & 0.06 & 0.72 & \textbf{0.02} & 0.69 & 0.4 \\ 
\sf{minor1} & 0.59 & 0.56 & error & error & error & \textbf{0.3} \\ 
\sf{minor2} & 8.9 & 61.3 & error & error & error & \textbf{0.38} \\ 
\sf{vardep} & 0.6 & 61.36 & error & error & error & \textbf{0.23} \\ 
\sf{minor3} & 1.51 & 0.94 & error & error & error & \textbf{0.34} \\ 
\hline
\end{tabu}
\caption{INV(AT) Benchmarks}
\label{tbl:at_benchmark}
\end{table}
\footnotetext{These 2 benchmarks do not have a valid solution, while ConcSolver, CVC4 and LoopInvGen have explicitly indicated this situation, the other 3 solvers hang infinitely. So we count the 3 solvers that come up with a no-solution response as successfully solved these benchmarks.}

\begin{table}
\fontsize{8}{10}\selectfont
\begin{tabu} to \textwidth {|X[c,12]|X[r,7]|X[r,4]|X[r,6]|X[r,8]|X[r,5]|X[r,7]|}
\hline
Benchmark & ConcSolver & CVC4 & EUSolver & Alchemist CS & ICE DT & LoopInvGen \\
\hline
\sf{anfp-new} & 1.06 & 3.39 & 103.98 & 21.36 & $\infty$ & \textbf{0.3}\\ 
\sf{anfp} & 1.43 & 3.39 & 42.85 & 13.2 & $\infty$ & \textbf{0.25} \\ 
\sf{array\_simp} & 2.01 & 0.14 & 20.22 & \textbf{0.09} & 0.78 & 1.5 \\ 
\sf{array\_vars} & 1.87 & \textbf{0.3} & 60.66 & 15.81 & 0.7 & 2.84 \\ 
\sf{array-new} & 1.03 & 0.13 & 6.25 & \textbf{0.08} & 171.72 & 0.37 \\ 
\sf{array} & 1.46 & 0.13 & 6.46 & \textbf{0.09} & 0.93 & 0.51 \\ 
\sf{cegar2\_vars-new} & 3.27 & 78.58 & $\infty$ & 1784.45 & 1.54 & \textbf{1.09} \\ 
\sf{cegar2\_vars} & 3.94 & 75.38 & $\infty$ & $\infty$ & \textbf{1.13} & 1.26 \\ 
\sf{cegar2-new} & 3.89 & 66.06 & 394.41 & 5.97 & \textbf{1.1} & \textbf{1.1} \\ 
\sf{cegar2} & 2.37 & 64.78 & 1202.13 & 13.52 & \textbf{1.09} & 1.26 \\ 
\sf{ex11\_simpl-new} & 0.64 & 0.07 & 6.26 & \textbf{0.05} & 0.73 & 3.09 \\ 
\sf{ex11\_simpl} & 1.4 & 1.89 & 740.53 & \textbf{0.15} & 0.75 & 0.3 \\ 
\sf{ex11\_vars-new} & 0.88 & 0.13 & 176.83 & \textbf{0.04} & 0.72 & 3.1 \\ 
\sf{ex11\_vars} & $\infty$ & error & $\infty$ & $\infty$ & $\infty$ & error \\ 
\sf{ex11-new} & 0.68 & 2.3 & $\infty$ & 14.8 & 1.46 & \textbf{0.24} \\ 
\sf{ex11} & 1.3 & 2.33 & error & \textbf{0.04} & 0.71 & 0.23 \\ 
\sf{ex14\_simp} & 0.69 & 0.1 & 0.8 & \textbf{0.02} & 0.67 & 1.39 \\ 
\sf{ex14\_vars} & 0.69 & 0.22 & 1.39 & \textbf{0.02} & 0.79 & 1.4 \\ 
\sf{ex14-new} & 0.59 & 0.11 & 0.77 & \textbf{0.03} & 0.84 & 0.33 \\ 
\sf{ex14} & 0.6 & 0.1 & 0.83 & \textbf{0.03} & 0.67 & 0.38 \\ 
\sf{fig1\_vars-new} & 0.82 & 2.66 & 59.49 & 25.39 & 9.8 & \textbf{0.4} \\ 
\sf{fig1\_vars} & 0.91 & 2.67 & 59.22 & 28.54 & 0.87 & \textbf{0.4} \\ 
\sf{fig1-new} & 0.69 & 0.58 & 5.33 & \textbf{0.23} & 7.76 & 0.39 \\ 
\sf{fig1} & 0.8 & 0.58 & 5.46 & \textbf{0.18} & 0.86 & 0.39 \\ 
\sf{fig3\_vars} & 1.07 & \textbf{0.06} & 163.15 & 16.85 & 0.73 & 0.33 \\ 
\sf{fig3} & 1.04 & \textbf{0.06} & 9.04 & 0.74 & 0.74 & 0.35 \\ 
\sf{fig9\_vars} & 0.66 & 0.12 & 37.12 & \textbf{0.04} & 0.69 & 0.25 \\ 
\sf{fig9} & 0.98 & \textbf{0.03} & 2.52 & 0.04 & 0.67 & 0.31 \\ 
\sf{formula22} & 0.39 & 60.08 & $\infty$ & 1.81 & 1.22 & \textbf{0.24} \\ 
\sf{formula25} & 0.39 & 60.08 & $\infty$ & $\infty$ & 1.24 & \textbf{0.3} \\ 
\sf{formula27} & 2.11 & 60.08 & $\infty$ & $\infty$ & 2.22 & \textbf{0.3 }\\ 
\sf{hola.05} & 1.9 & 15.17 & 14.35 & 2.43 & \textbf{0.83} & 0.97 \\ 
\sf{hola.07} & $\infty$ & $\infty$ & $\infty$ & $\infty$ & $\infty$ & \textbf{2.24} \\ 
\sf{hola.20} & $\infty$ & 60.4 & $\infty$ & $\infty$ & $\infty$ & \textbf{1.74} \\ 
\sf{hola.41} & 62.53 & 1400.93 & $\infty$ & $\infty$ & 1.35 & \textbf{0.59} \\ 
\sf{hola.44} & 40.61 & 56.26 & error & \textbf{0.23} & 0.91 & 2.11 \\ 
\sf{matrix2\_simp} & $\infty$ & error & $\infty$ & $\infty$ & $\infty$ & error \\ 
\sf{matrix2} & $\infty$ & error & $\infty$ & $\infty$ & $\infty$ & $\infty$ \\ 
\sf{treax1} & 0.49 & 0.09 & 0.8 & \textbf{0.02} & 0.63 & 1.14 \\ 
\sf{trex1\_vars} & 0.56 & 0.21 & 1.12 & \textbf{0.02} & 0.63 & 1.14 \\ 
\sf{trex3\_vars} & $\infty$ & error & $\infty$ & $\infty$ & $\infty$ & error \\ 
\sf{trex3} & $\infty$ & error & $\infty$ & $\infty$ & $\infty$ & error \\ 
\sf{vsend} & 0.54 & 0.21 & 0.8 & \textbf{0.02} & 1.13 & 3.08 \\ 
\hline
\end{tabu}
\caption{INV(non-AT) Benchmarks}
\label{tbl:at_benchmark}
\end{table}

\end{document}